\documentclass[letterpaper, 10 pt, conference]{ieeeconf}  

\IEEEoverridecommandlockouts                              

\usepackage{hyperref}
\usepackage{graphicx}		
\usepackage{wrapfig}
\usepackage[format=plain,font=footnotesize,labelfont=bf,labelsep=period]{caption}
\usepackage{sidecap} 
\usepackage[export]{adjustbox}
\usepackage{subcaption}

\usepackage{amsmath}
\usepackage{amssymb}
\usepackage{mathtools}
\usepackage{algorithm}
\usepackage{algpseudocode}

\usepackage{pdfsync}
\usepackage[normalem]{ulem}
\usepackage{paralist}	
\usepackage[space]{grffile} 

\usepackage{color}
\usepackage{dirtytalk}

\usepackage{amsthm}

\newtheorem{theorem}{Theorem}

\theoremstyle{definition}
\newtheorem{definition}{Definition}
\theoremstyle{remark}

\theoremstyle{definition}

\theoremstyle{definition}

\newcommand{\R}{\mathbb{R}}

\newcommand{\C}{\mathcal{C}}

\definecolor{darkblue}{RGB}{0,0,102}
\definecolor{lightblue}{RGB}{77,77,148}

\definecolor{gold}{RGB}{234, 170, 0}
\definecolor{metallic_gold}{RGB}{139, 111, 78}

\newcommand{\mb}[1]{\mathbf{ #1 }}
\newcommand{\bs}[1]{\boldsymbol{ #1 }}

\DeclareMathOperator*{\argmin}{argmin}

\newcommand{\derp}[2]{\frac{\partial #1 }{\partial #2 }}
\newcommand{\x}{\mathbf{x}}
\newcommand{\lmat}{\begin{bmatrix}}
\newcommand{\rmat}{\end{bmatrix}}
\newcommand{\hphi}{\overline{h}}

\title{\LARGE \bf
 Measurement-Robust Control Barrier Functions: 
\\ Certainty in Safety with Uncertainty in State 
}

\author{Ryan K. Cosner, Andrew W. Singletary, Andrew J. Taylor, Tamas G. Molnar, \\ Katherine L. Bouman, and Aaron D. Ames
\thanks{This research is supported in part by the National Science Foundation, CPS Award \#1932091, as well as Aerovironment. A.J. Taylor is supported by DARPA award HR00111890035.}
\thanks{
R. K. Cosner, A. W. Singletary, T. G. Molnar, and A. D. Ames are with the Department of Mechanical and Civil Engineering, California Institute of Technology, Pasadena, CA 91125, USA, {\tt\small \{rkcosner, asinglet, tmolnar, ames\}@caltech.edu}.}
\thanks{
A. J. Taylor and K. L. Bouman are with Computing and Mathematical Sciences, California Institute of Technology, Pasadena, CA 91125, USA, {\tt\small \{ajtaylor, klbouman\}@caltech.edu}.}
}

\begin{document}

\maketitle
\thispagestyle{empty}
\pagestyle{empty}

\begin{abstract}
The increasing complexity of modern robotic systems and the environments they operate in necessitates the formal consideration of safety in the presence of imperfect measurements. In this paper we propose a rigorous framework for safety-critical control of systems with erroneous state estimates. We develop this framework by leveraging Control Barrier Functions (CBFs) and unifying the method of Backup Sets for synthesizing control invariant sets with robustness requirements---the end result is the synthesis of \emph{Measurement-Robust Control Barrier Functions (MR-CBFs)}.
This provides theoretical guarantees on safe behavior in the presence of imperfect measurements and improved robustness over standard CBF approaches. We demonstrate the efficacy of this framework both in simulation and experimentally on a Segway platform using an onboard stereo-vision camera for state estimation.


\end{abstract}

\section{INTRODUCTION}

Safety is of utmost importance in many modern control applications, including autonomous vehicles, medical and industrial robotics \cite{knight2002safety}.
The growing complexity of these systems demands that safety properties are rigorously encoded in the controller design.
Such systems are typically described as safe if their state never leaves a prescribed \textit{safe set},
and
\textit{Control Barrier Functions (CBFs)} \cite{ames2017control, borrmann2015control} have become increasingly popular \cite{ames2019control,nguyen2016exponential} as a tool for achieving
safety.
In this paper, we focus on two challenges related to safety-critical control realized via CBFs: finding admissible inputs and making these inputs robust to uncertainty.

The first challenge is guaranteeing that a safe control input is always available.
If safe control actions exist---i.e., satisfy input constraints---over the entire safe set, the set is called \textit{control invariant} \cite{aubin2011viability}.  Yet control invariance is not guaranteed in general---safe actions may not exist for all points in a given safe set.
Therefore, identifying control invariant sets is critically important for implementing safety-critical controllers in robotic systems. 
Hamilton-Jacobi reachability analysis can be performed to compute such sets \cite{bansal2017hamilton}, but is intractable for high dimensional systems.
Here we adapt the method of  \textit{Backup Sets} introduced in \cite{gurriet2020scalable} as a computationally tractable way of achieving control invariance.

\begin{figure}[!ht]
\centering 
    \vspace{0.65cm}
    \includegraphics[width=\linewidth]{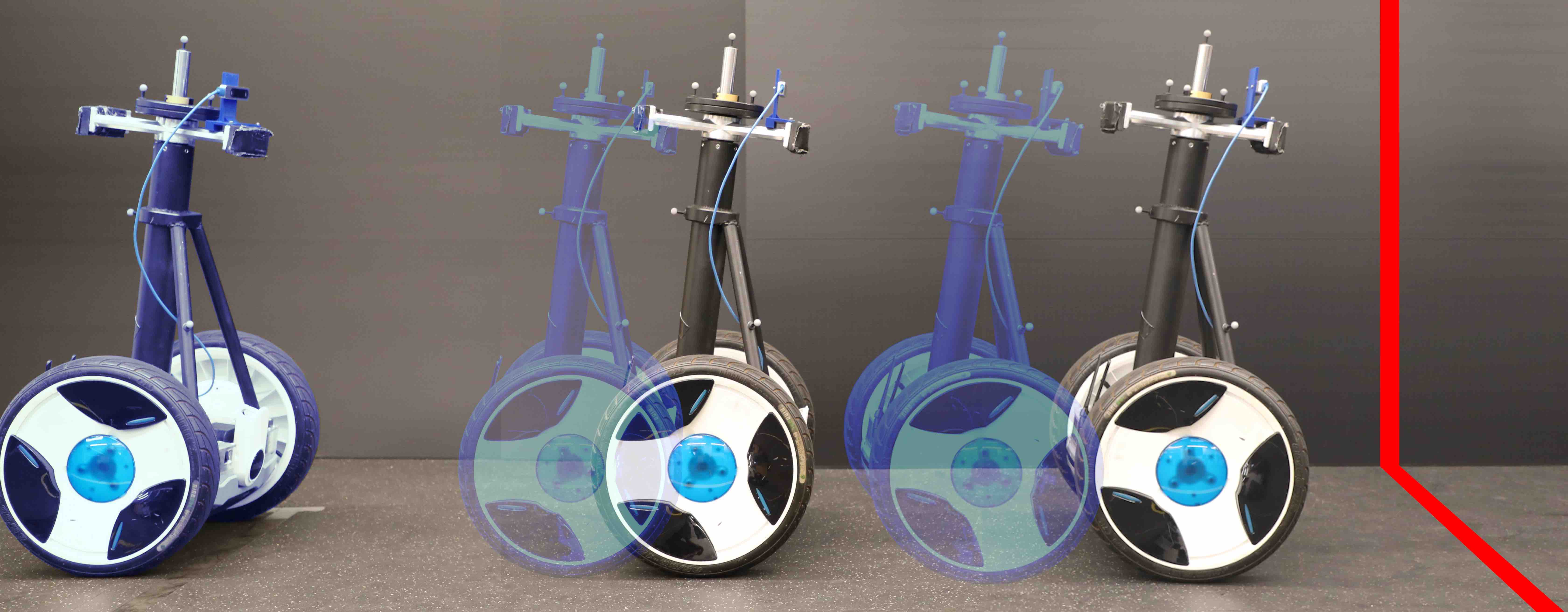}
    \caption{Visualization of desired Segway behavior. The Segway is driven from left to right and must not cross the red line. The transparent blue images of the Segway represent the measured position whereas the opaque images represent the true position. Traditional CBFs do not account for this uncertainty.}
    \label{fig:segway}
    \vspace{-0.5cm}
\end{figure}

The second challenge is that controllers rely on state measurements that are often imperfect or uncertain---especially for dynamic robotic systems.
This can cause unsafe behavior if not accounted for in the control design and, as such, has been addressed from multiple perspectives. 
The work in \cite{jankovic2018robust,takano2018robust} considers robust CBF formulations with worst-case disturbance bounds to achieve safety.
Safety guarantees in the presence of measurement noise are addressed from a stochastic perspective in \cite{clark2019control, nilsson2020lyapunov}.
Controllers robust to state estimation errors were proposed for sampled-data-systems via an interval-arithmetic condition in \cite{singletary2020control} and for continuous systems via estimate-error bounding in \cite{dean2020guaranteeing}.
In \cite{dean2020guaranteeing} safety and robustness
were enforced
by \textit{Measurement-Robust Control Barrier Functions (MR-CBFs)}.
This approach was inspired by vision-based control \cite{codevilla2018end,lambert2018deep,tang2018aggressive}, where state information is observed through a complex transformation.




This paper presents a safety-critical control framework that allows for the synthesis of control invariant sets that are robust to measurement uncertainty, all with a view toward experimental realization.  
The main contributions of this work are twofold.
Firstly, we integrate the method of Backup Sets for ensuring control invariance \cite{gurriet2020scalable} with the framework of MR-CBFs \cite{dean2020guaranteeing}.
This leads to practically achievable safety guarantees even in the presence of measurement uncertainty, establishing measurement-robust safety-critical control.
Secondly, we present the first experimental demonstration of both MR-CBFs and the proposed method
by controlling the motion of a Segway using camera data.
The experiments validate the robust safety guarantees provided by our method.






\section{Preliminaries} \label{sec:background}

First we provide a review of safety-critical control through Control Barrier Functions (CBFs) and synthesis of control invariant sets via the Backup Set method. 

\subsection{Control Barrier Functions}

Consider the nonlinear control affine system given by: 
\begin{equation}
    \dot{\x} = \mb{f}(\x) + \mb{g}(\x) \mb{u}, \quad \x \in \R^n , \;\mb{u} \in \R^m, \label{eq:dynamics}
\end{equation}
where ${\mb{f}: \R^n \to \R^n}$ and ${\mb{g}: \R^n \to \R^{n\times m}}$ are locally Lipschitz continuous.
Given a locally Lipschitz continuous controller ${\mb{k}: \R^n \to \R^m}$,  the closed-loop dynamics are: 
\begin{equation}
    \dot{\x} = \mb{f}_\textrm{cl}(\x) = \mb{f}(\x) + \mb{g}(\x)\mb{k}(\x), \label{eq:cl}
\end{equation}
where $\mathbf{f}_\textrm{cl}:\R^n\to\R^n$ is also locally Lipschitz continuous.
Therefore, for any initial condition ${\x(0) = \x_0 \in \R^n}$ there exists an interval $I(\x_0) \triangleq [0,t_\textrm{max})$ such that $\x(t)$ is the unique solution to \eqref{eq:cl} for $t\in I(\x_0)$ \cite{perko2013differential}.
Throughout this paper we assume $I(\mb{x}_0) = [0, \infty)$.

The notion of safety is formalized by defining a \textit{safe set} $\mathcal{C} \subset \R^n $ 
in the state space that the system must remain within.
In particular, consider the set $\mathcal{C}$ as the $0$-superlevel set of a continuously differentiable function $h : \R^n \to \R$: 
\begin{align}
\begin{split}
    \mathcal{C} & \triangleq \{ \x\in \R^n : h(\x) \geq 0 \}, \\
    \partial \mathcal{C} & \triangleq \{ \x \in \R^n : h(\x) = 0 \}, \\
    \textrm{Int}(\mathcal{C}) & \triangleq \{ \x \in \R^n : h(\x) > 0 \}.
\end{split}
\label{eq:safe_set}
\end{align}
We assume that zero is a regular value of $h$ and $\mathcal{C}$ is non-empty and has no isolated points, that is, $h(\x) = 0 \implies \frac{\partial h }{\partial \x}(\x)  \neq 0  $, $\textrm{Int}(\mathcal{C})\neq \emptyset $, and $\overline{\textrm{Int}(\mathcal{C})} = \mathcal{C}$. 
In this context, safety is synonymous with the forward invariance of $\mathcal{C}$:

\begin{definition}[\textit{Forward Invariance and Safety}]
    A set $\mathcal{C}\subset \R^n$ is \textit{forward invariant} if for every $\x_0 \in \mathcal{C}$, the solution to (\ref{eq:cl}) satisfies $\mb{x}(t) \in \mathcal{C}$ for all $t\geq 0$. The closed-loop system \eqref{eq:cl} is \textit{safe} with respect to set $\mathcal{C}$ if $\mathcal{C}$ is forward invariant.
\end{definition}

We call a continuous function ${\alpha: \R \to \R}$ as extended class-$\mathcal{K}_{\infty}$ ($\mathcal{K}_{\infty,e}$) if it is strictly monotonically increasing and satisfies ${\alpha(0) = 0}$, ${\lim_{r \to -\infty}\alpha(r) = -\infty}$, and ${\lim_{r \to \infty} \alpha (r) = \infty}$. Control Barrier Functions (CBF) can be used to synthesize controllers ensuring the safety of the closed-loop system \eqref{eq:cl} with respect to a given set $\mathcal{C}$. 
\begin{definition}[\textit{Control Barrier Function (CBF)}, \cite{ames2017control}]\label{def:cbf}
Let ${\C\subset\R^n}$ be a safe set given by~(\ref{eq:safe_set}).
The function $h$ is a \textit{Control Barrier Function} (CBF) for \eqref{eq:dynamics} on $\mathcal{C}$ if there exists $\alpha\in\mathcal{K}_{\infty,e}$ such that for all $\mb{x}\in\C$:
\begin{equation}
\label{eqn:cbf}
     \sup_{\mb{u}\in\R^m} \dot{h}(\mb{x},\mb{u}) \triangleq \underbrace{\derp{h}{\mb{x}}(\mb{x})\mb{f}(\mb{x})}_{L_\mb{f}h(\mb{x})}+\underbrace{\derp{h}{\mb{x}}(\mb{x})\mb{g}(\mb{x})}_{L_\mb{g}h(\mb{x})}\mb{u}\geq-\alpha(h(\mb{x})),
\end{equation}
where $L_\mb{f}h: \R^n \to \R$ and $L_\mb{g}h: \R^n \to \R^m $ are the Lie derivatives of $h$ with respect to $\mb{f}$ and $\mb{g}$, respectively.
\end{definition}

Intuitively, the CBF constraint \eqref{eqn:cbf} requires a system to slow down as it approaches the boundary of the safe set (the right-hand side of \eqref{eqn:cbf} increases to $0$ as the value of $h$ approaches $0$). A main result in \cite{ames2014control, xu2015robustness} relates CBFs
to the safety of the closed-loop system \eqref{eq:cl} with respect to $\C$:
\begin{theorem}\label{thm:cbf_safe}
Given a safe set $\C\subset\R^n$,
if $h$ is a CBF for \eqref{eq:dynamics} on $\C$, then any locally Lipschitz continuous controller $\mb{k}:\R^n\to\R^m$
satisfying
\begin{equation}
L_\mb{f}h(\mb{x})+L_\mb{g}h(\mb{x})\mb{k}(\mb{x})\geq-\alpha(h(\mb{x}))
\end{equation}
for all $\mb{x}\in\C$, renders the system \eqref{eq:cl} safe w.r.t. $\C$.
\end{theorem}

Given a nominal (but not necessarily safe) locally Lipschitz continuous controller $\mb{k}_d:\R^n\to\R^m$ and a CBF $h$, the CBF-Quadratic Program \eqref{eqn:CBF-QP} \cite{ames2017control} ensures safety:
\begin{align}
\label{eqn:CBF-QP}
\tag{CBF-QP}
\mb{k}(\mb{x}) =  \,\,\underset{\mb{u} \in \R^m}{\argmin}  &  \quad \frac{1}{2} \| \mb{u} -\mb{k}_d(\mb{x})\|_2^2  \\
\mathrm{s.t.} \quad & \quad L_\mb{f}h(\mb{x}) + L_\mb{g}h(\mb{x}) \mb{u} \geq - \alpha (h(\mb{x})). \nonumber
\end{align}

\subsection{Generating Control Invariant Sets via Backup Sets}



To guarantee that a safe control action exists, one needs to ensure the existence of a function $h$ satisfying the CBF condition \eqref{eqn:cbf}.
For a given safe-set $\C$, fulfilling this requirement can be nontrivial and potentially impossible. To this end, we restrict our focus to a set $\mathcal{C}_I \subseteq \mathcal{C}$ which is control invariant:  


%
%

\begin{definition}[\textit{Control Invariance}]
    A set $\mathcal{C}_I\subseteq\C$ is \textit{control invariant} if there exists a controller $\mb{k}: \R^n \to \R^m$ such that  $\mathcal{C}_I$ is forward invariant with respect to the system \eqref{eq:cl}. 
\end{definition}




While directly computing control invariant sets remains challenging in general, we may define one implicitly via a backup set \cite{gurriet2020scalable}. Consider a desired safe set $\C\subset\R^n$,
which is not necessarily control invariant. Suppose there exists a set $\mathcal{C}_B\subset\C$, defined as the $0$-superlevel set of a continuously differentiable function $h_B:\R^n\to\R$, which is known \textit{a priori} to be control invariant and can be rendered forward invariant by a known locally Lipschitz continuous \textit{backup controller} $\mb{k}_B:\R^n\to\R^m$.
We refer to $\C_B$ as the \textit{backup set}. 
For simple backup controllers (such as linear state feedback controllers designed for the linearization of a system) it is possible to find analytical expressions for local regions of attraction to serve as backup sets. Alternatively, numerical tools such as Sums-of-Squares (SOS) may be used to synthesize control invariant 
sets \cite{tan2006stability}. 

We extend the backup set to a larger control invariant set ${\C_I\subset\R^n}$, satisfying ${\C_B\subseteq\C_I\subseteq\C}$, by considering the \textit{backup trajectory} over a finite and fixed time  ${T\in\R_{>0}}$  as follows.
By assumption, for any $\mb{x}\in\R^n$ there exists a unique solution $\bs{\varphi}:[0,T]\to\R^n$ satisfying:
\begin{align}
\begin{split}
    \frac{\mathrm{d}}{\mathrm{d}\tau}\bs{\varphi}(\tau) & = \mb{f}(\bs{\varphi}(\tau)) + \mb{g}(\bs{\varphi}(\tau))\mb{k}_B(\bs{\varphi}(\tau)), \\
    \bs{\varphi}(0) & = \mb{x}.
    \end{split}
\end{align}
The solution $\bs{\varphi}$ may be interpreted as the evolution of the system over the interval $[0,T]$ from a state, $\mb{x}$, under the backup controller $\mb{k}_B$. In particular, the current state $\mb{x}(t)$ may be used as initial condition in specifying $\bs{\varphi}$. We denote $\phi_{\tau}^{\mb{k}_b}(\mb{x}) \triangleq \bs{\varphi}(\tau)$ for the initial condition $\mb{x}$.

Using this notation, we may define the set $\C_I\subseteq\C$ as:
\begin{equation} 
\mathcal{C}_I = \left\{\begin{array}{l|c}
        & h(\phi_\tau^{\mb{k}_B}(\mb{x})) \geq 0, \forall \tau \in [0, T ]\\
        \mb{x} \in \C &  \textrm{and} \\
        &  h_B(\phi_T^{\mb{k}_B}(\mb{x})) \geq 0 
        \end{array}\right\}. \label{eq:C_I}
\end{equation}
The first inequality implies safety under the backup policy (${\phi^{\mb{k}_B}_\tau(\mb{x})\in\mathcal{C}}$ for all ${\tau \in [0, T]}$), and the second inequality implies the backup trajectory reaches $\mathcal{C}_B$ by time $T$ (${\phi^{\mb{k}_B}_T(\mb{x})\in\mathcal{C}_B}$). The set $\mathcal{C}_I$ is thus control invariant as there exists at least one controller, $\mb{k}_B$, which renders it forward invariant.  
While $\mathcal{C}_I$ is not necessarily the largest control invariant subset of $\mathcal{C}$
(see \textit{viability kernel}, \cite{aubin2011viability}),
the backup sets provide a computationally tractable method for finding an under-approximation of the largest control invariant set.

For notational simplicity, we define the continuously differentiable functions $\hphi_\tau:\R^n\to\R$ and $\hphi_B:\R^n\to\R$ as:
\begin{align}
    \hphi_\tau(\mb{x}) & \triangleq h(\phi_\tau^{\mb{k}_B}(\mb{x})), & \hphi_B(\mb{x}) & \triangleq h_B(\phi_T^{\mb{k}_B}(\mb{x})).
\end{align}
Given these definitions, the CBF condition \eqref{eqn:cbf} can then be specified for the set $\mathcal{C}_I$ at a point $\mb{x}\in\C_I$ as follows:
\begin{align}
\begin{split}
    L_\mb{f}\hphi_\tau (\mb{x})  + L_\mb{g}\hphi_\tau (\mb{x})\mb{u} &\geq - \alpha (\hphi_\tau(\mb{x}) ), \quad \forall \tau \in [0, T], \\
    L_\mb{f}\hphi_B (\mb{x})  + L_\mb{g}\hphi_B (\mb{x}) \mb{u}&\geq - \alpha (\hphi_B(\mb{x}) ).
\end{split}
\label{eq:backup_set_const}
\end{align}
Any locally Lipschitz continuous controller that takes values satisfying \eqref{eq:backup_set_const} for all $\mb{x}\in \C_I$ will keep the closed loop system \eqref{eq:cl} safe with respect to $\mathcal{C}_I$; see \cite[p. 6]{gurriet2018online}.

We note that enforcing the first constraint in~(\ref{eq:backup_set_const}) is not necessarily tractable as it must hold for all ${\tau \in [0, T]}$. To resolve this, it can be reduced to a finite collection of more conservative constraints through constraint tightening. A controller which implements the finite number of tightened constraints, and thus renders \eqref{eq:cl} safe with respect to $\C_I$, is given by the Backup Set Quadratic Program (BS-QP): 
\begin{align}
\label{eqn:BS-QP}
\tag{BS-QP}
\mb{k}(\mb{x}) =   \,\,\underset{\mb{u} \in \R^m}{\argmin}    &\quad  \frac{1}{2} \| \mb{u} -\mb{k}_d(\mb{x})\|_2^2  \nonumber \\
\mathrm{s.t.} \quad & L_\mb{f}\hphi_{\tau_j} (\mb{x})  + L_\mb{g}\hphi_{\tau_j} (\mb{x})\mb{u} \geq - \alpha ({h}_{\tau_j}(\mb{x}) - \mu), \nonumber\\
& L_\mb{f}\hphi_B (\mb{x})  + L_\mb{g}\hphi_B (\mb{x})\mb{u} \geq - \alpha ({h}_B(\mb{x})),  \nonumber
\end{align}
for all ${\tau_j \in \{ 0, \Delta_t, \dots, T\}}$, where ${\Delta_t \in \R_{>0}}$ is a time-step such that ${T/\Delta_t \in \mathbb{N}}$,
%
%
while $\mu \in \R_{>0}$ satisfies:
\begin{equation}
\mu\geq \frac{\Delta_t}{2} \mathfrak{L}_h \sup_{\mb{x}\in \mathcal{C}}\Vert \mb{f}(\mb{x}) + \mb{g}(\mb{x})\mb{k}_B(\mb{x})\Vert_2, \label{eq:varep}
\end{equation}
with ${\mathfrak{L}_h\in\R_{>0}}$ a Lipschitz constant for $h$ on $\C$ \cite[Thm. 1]{gurriet2020scalable}.




\begin{figure*}[!ht]
\centering 
    \includegraphics[width=0.23\linewidth]{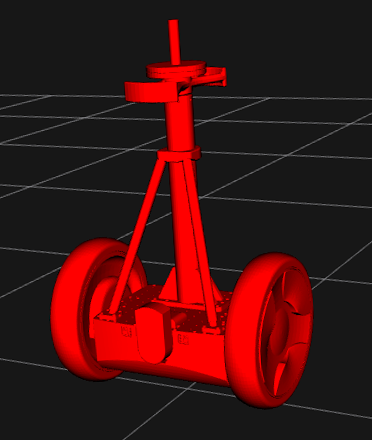}       \includegraphics[height=1.9in, width=0.37\linewidth]{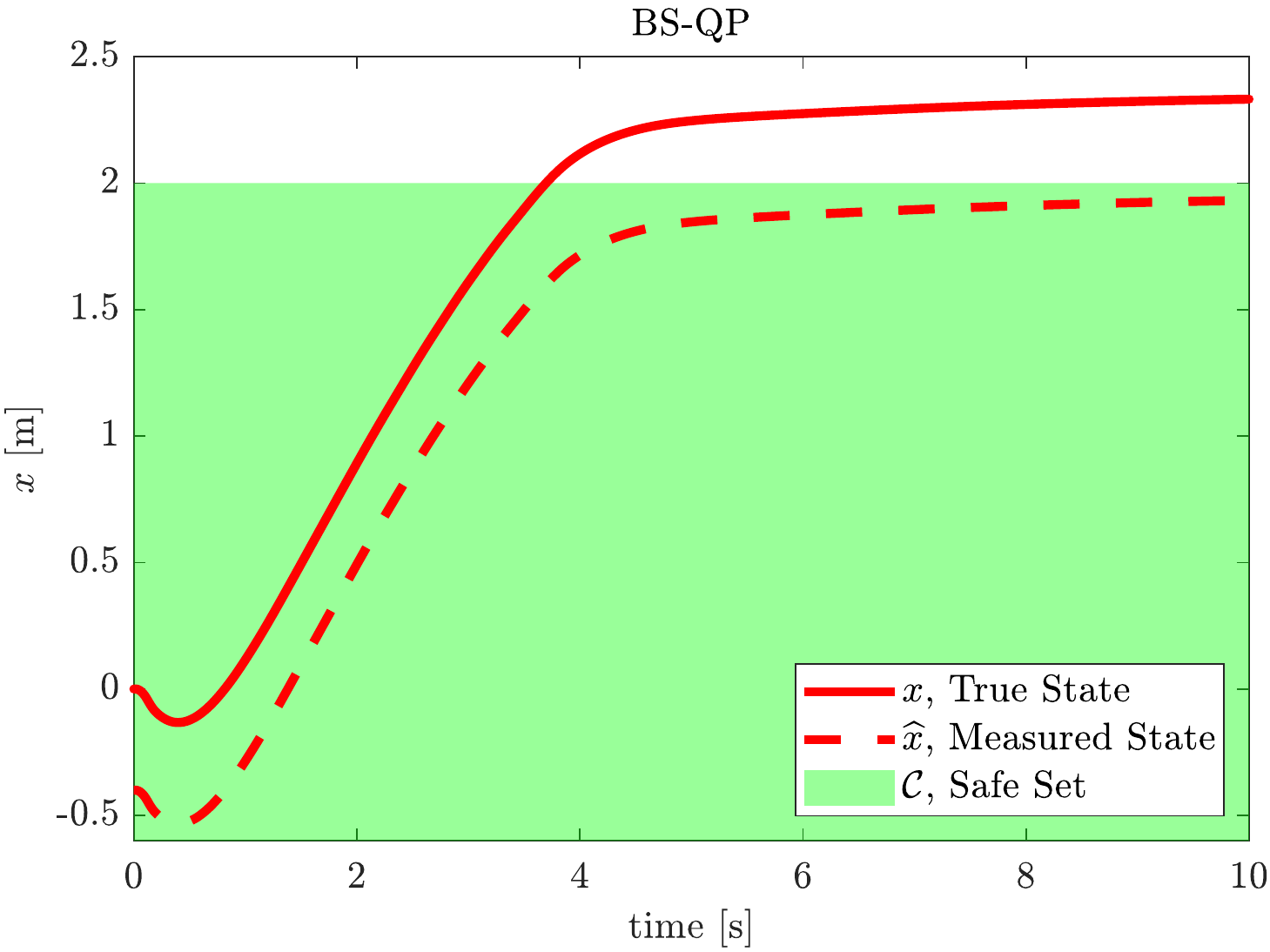}    \includegraphics[height=1.9in,width=0.37\linewidth]{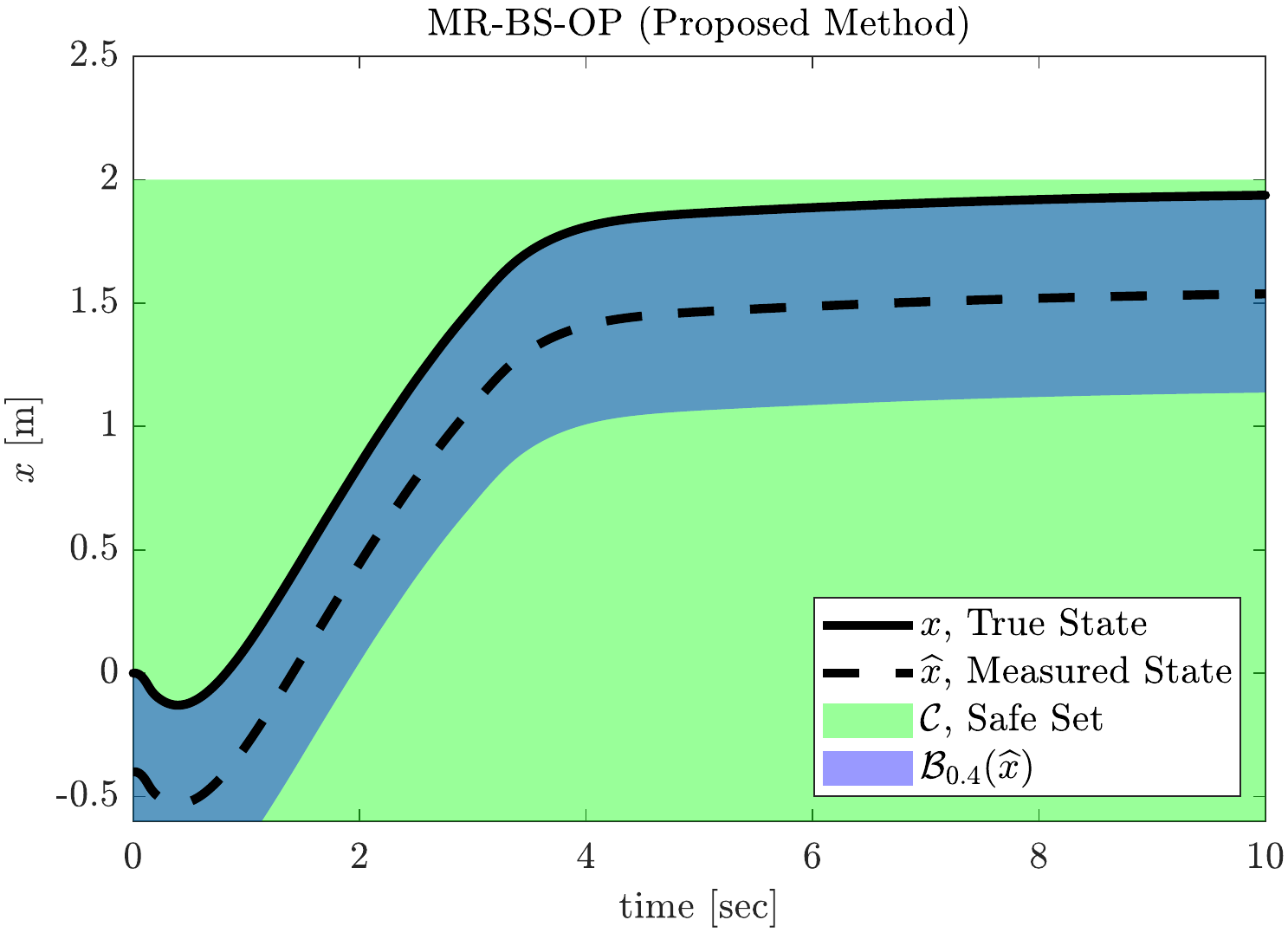}
    \caption{Simulation results for a measurement model of $\widehat{x} = x - 0.4 \textrm{ m}$ and constant desired velocity of 1 $\textrm{m/s}$. \textbf{(Left)} An image of the simulated Segway model. \textbf{(Center)} Trajectories generated using the \ref{eqn:BS-QP}.
    Solid line represents the true state, dashed line shows the estimated state, and green region indicates the safe set $\mathcal{C}$. The true trajectory fails to be safe and exits the safe set at $t=3$ s. \textbf{(Right)} Trajectories generated using the \ref{eq:MRBS-OP}.
    An additional robustness region is plotted in blue to indicate the set of of true states which the control input renders safe. Both the true and measured trajectories are safe demonstrating the robustness of the \ref{eq:MRBS-OP} when compared to the \ref{eqn:BS-QP}.}
    \label{fig:sim} 
    \vspace{-0.5cm}
\end{figure*}
\section{Measurement Robustness }

\label{sec:MR}
The guarantees endowed by the above controllers require perfect knowledge of the state $\mb{x}$, which often is an unrealistic assumption in practice. In particular, the relationship between the state of the system and the measurements, such as images or point clouds, can be complex and not fully known \cite{codevilla2018end, lambert2018deep, tang2018aggressive}. In this section we revisit measurement-model uncertainty and present our main result in the form of a measurement-robust version of the \ref{eqn:BS-QP}. 

\subsection{Measurement-Model Uncertainty}
To achieve robustness, we consider a structured form of measurement-model uncertainty that modifies the CBF condition \eqref{eqn:cbf} \cite{dean2020guaranteeing}. We assume the state $\mb{x}$ is not directly available, but rather a state-dependent sensor measurement: 
\begin{equation}
    \label{eq:measure}
    \mb{y} = \mb{p}(\mb{x}),
\end{equation}
where $\mb{y}\in\R^k$ and ${\mb{p}: \R^n \to \R^k}$ is locally Lipschitz continuous.
An estimate of the state, $\widehat{\mb{x}}\in\R^n$, is reconstructed from $\mb{y}$ (such as through measurement models or data-driven methods \cite{lambert2018deep,tang2018aggressive}). In particular, we assume the map from measurements to state estimates is imperfect (does not recover the true state exactly), and is given by the locally Lipschitz continuous function $\widehat{\mb{q}}: \R^k \to \R^n$ as follows:
\begin{equation}
\label{eqn:state_error}
    \widehat{\mb{x}} \triangleq \widehat{\mb{q}}(\mb{y}) = \mb{x}+\mb{e}(\mb{x}),
\end{equation}
where the state error function ${\mb{e}:\R^n\to\R^n}$ is unknown and implicitly defined by $\widehat{\mb{q}}$. 

The error function $\mb{e}$ can often be characterized via upper bounds on measurement-model uncertainty. In particular, we assume that while the state error $\mb{e}(\mb{x})$ is not known for a given state $\mb{x}\in\R^n$, it is within a compact error set ${\mathcal{E}(\mb{y})}$ specified by a set-valued function ${\mathcal{E}:\R^k\to\mathcal{P}(\R^n)}$
($\mathcal{P}$ denotes the power set), that is,
we have ${\mb{e}(\mb{x}) \in \mathcal{E}(\mb{y}) = \mathcal{E}(\mb{p}(\mb{x}))}$.
%
%
The error set can be conservatively characterized via the function ${\epsilon:\R^k\to\R_{\geq 0}}$ defined as:
\begin{equation}
    \epsilon(\mb{y}) \triangleq \max_{\mb{e} \in \mathcal{E}(\mb{y}) }\Vert \mb{e} \Vert_2. \label{eq:epsilon}
\end{equation}

Since the controller only has access to the measurement and the state estimate, systems with measurement-model uncertainty evolve according to: 
\begin{equation}
    \dot{\mb{x}} = \mb{f}(\mb{x}) + \mb{g}(\mb{x}) \mb{k}(\mb{y}, \widehat{\mb{x}}). \label{eq:mu_ode}
\end{equation}
The error bound can be used to synthesize controllers which render such systems provably safe as follows \cite{dean2020guaranteeing}:
\begin{theorem}
Given a safe set $\C\subset\R^n$, assume that $L_\mb{f}h$, $L_\mb{g}h$, and ${\alpha \circ h}$ are Lipschitz continuous on $\C$ with Lipschitz constants $\mathfrak{L}_{L_\mb{f}h}, \mathfrak{L}_{L_\mb{g}h}$, and ${\mathfrak{L}_{\alpha \circ h }\in\R_{\geq 0}}$, respectively. Define the function ${\epsilon:\R^k\to\R_{\geq 0}}$ as in \eqref{eq:epsilon}, and define the functions ${a,b:\R^k\to\R_{\geq 0}}$ as ${a(\mb{y}) = (\mathfrak{L}_{L_\mb{f}h} + \mathfrak{L}_{\alpha \circ h})\epsilon(\mb{y})}$ and ${b(\mb{y}) = \mathfrak{L}_{L_\mb{g}h}\epsilon (\mb{y})}$. If $\mb{k}: \R^k \times \R^n \to \R^m$ is a Lipschitz continuous controller satisfying: 
\begin{multline}
    L_\mb{f}h(\widehat{\mb{x}}) + L_\mb{g}h(\widehat{\mb{x}}) \mb{k}(\mb{y}, \widehat{\mb{x}}) \label{eq:mrcbf_condition} \\ 
    - (a(\mb{y}) + b(\mb{y})\Vert \mb{k}(\mb{y}, \widehat{\mb{x}}) \Vert_2) \geq - \alpha (h(\widehat{\mb{x}}))
\end{multline}
for all ${\mb{x} \in \mathcal{C}}$, with ${\mb{y}=\mb{p}(\mb{x})}$ and ${\widehat{\mb{x}} = \widehat{\mb{q}}(\mb{y})}$, then the system \eqref{eq:mu_ode} is safe with respect to $\C$. 
\end{theorem}
\noindent A continuously differentiable function ${h:\R^n \to \R}$ for which such a controller exists is termed a Measurement-Robust Control Barrier Function (MR-CBF) \cite{dean2020guaranteeing}. As compared to the original CBF constraint \eqref{eqn:cbf}, the MR-CBF constraint \eqref{eq:mrcbf_condition} adds additional terms incorporating bounds on the measurement error that ensure the system is safe to all possible states following from a given measurement. The original CBF constraint is recovered in the absence of measurement error ($\epsilon(\mb{y}) = 0$).

\subsection{Measurement-Robust Backup Set Optimization Program}
\label{sec:combine}
In this section we present our main result in the form of a safety-critical control paradigm that is robust to measurement uncertainty. This is accomplished by unifying the Backup Set method with MR-CBFs, using the MR-CBF condition \eqref{eq:mrcbf_condition} the finite set of constraints imposed in the \ref{eqn:BS-QP} become:
\begin{align}
\begin{split}
      L_\mb{f}\hphi_{\tau_j} & (\widehat{\mb{x}}) + L_\mb{g}\hphi_{\tau_j} (\widehat{\mb{x}})\mb{u}\\
      & - (a_{\tau_j}(\mb{y}) + b_{\tau_j}(\mb{y})\Vert \mb{u} \Vert_2)
      \geq -\alpha(\hphi_{\tau_j} (\widehat{\mb{x}}) - \mu), \\
      L_\mb{f}\hphi_B & (\widehat{\mb{x}}) + L_\mb{g}\hphi_B (\widehat{\mb{x}})\mb{u} \\
      & - (a_B(\mb{y}) + b_B(\mb{y})\Vert \mb{u} \Vert_2)
      \geq -\alpha(\hphi_B (\widehat{\mb{x}})),
\end{split}
\label{eq:MRCBF_SDS}
\end{align}
with parameter functions: 
\begin{align}
\begin{split}
    a_{\tau_j}(\mb{y}) &= (\mathfrak{L}_{L_\mb{f}\hphi_{\tau_j}} + \mathfrak{L}_{\alpha} \mathfrak{L}_{\hphi_{\tau_j}}) \epsilon(\mb{y}), \quad
    b_{\tau_j}(\mb{y}) = \mathfrak{L}_{L_\mb{g}{\hphi_{\tau_j}}}\epsilon(\mb{y}), \\
    a_B(\mb{y}) &= (\mathfrak{L}_{L_\mb{f}\hphi_B} +  \mathfrak{L}_{\alpha}\mathfrak{L}_{\hphi_B}) \epsilon (\mb{y}), \quad
    b_B(\mb{y}) = \mathfrak{L}_{L_\mb{g}\hphi_B} \epsilon(\mb{y}),
\end{split}
\label{eq:parameter_funcs}
\end{align}
for all ${\tau_j \in \{0, \Delta_t, \dots, T\}}$, with $\epsilon(\mb{y}) $ defined as in \eqref{eq:epsilon} 
and $\mathfrak{L}$ represents the Lipschitz constant of its subscripted function on $\R^n$. These constructions enable the following definition:
\begin{definition}[\textit{Measurement-Robust Implicit Safe Set}]
The set $\mathcal{C}_I \subseteq \mathcal{C} \subseteq \R^n $ defined as in \eqref{eq:C_I} is a \textit{Measurement-Robust Implicit Safe Set} (MRISS) for the error bound $\epsilon:\R^k\to\R_{\geq 0}$ with parameter functions $(a_0,b_0, \dots, a_{\Delta_t}, b_{\Delta_t},a_B,b_B):\R^k\to\R_{\geq 0}$  if: 
\begin{itemize}
    \item the functions $\{ \hphi_{0}, \hphi_{\Delta_t}, \dots, \hphi_T, \hphi_B\} $, their Lie derivatives, and $\alpha$ are Lipschitz continuous on $\C_I$, 
    \item the constant $\mu \in \R_{\geq 0 }$ satisfies~(\ref{eq:varep}), 
    \item and for all ${\mb{x} \in \mathcal{C}_I}$ 
    there exists $\mb{u} \in \R^m$ satisfying \eqref{eq:MRCBF_SDS}. 
\end{itemize}
\end{definition}
Next, using this definition, we show that the safety of such sets can be made robust to measurement model uncertainty. 
\begin{theorem} 
\label{thm:main}
    Given a MRISS $\mathcal{C}_I$, if ${\mb{k}: \R^k \times \R^n \to \R^m}$ is a Lipschitz continuous controller that satisfies \eqref{eq:MRCBF_SDS} with parameter functions \eqref{eq:parameter_funcs} for all $ {\mb{x} \in \mathcal{C}_I } $ 
    with ${\mb{y}=\mb{p}(\mb{x})}$ and ${\widehat{\mb{x}} = \widehat{\mb{q}}(\mb{y})}$,
    then system \eqref{eq:mu_ode} is safe with with respect to $\mathcal{C}_I$. 
    
\end{theorem}

\begin{proof}
    For any function $\hphi \in \{ \hphi_{0}, \hphi_{\Delta_t}, \dots, \hphi_{T}, \hphi_B \} $ let
    \begin{equation}
        c(\mb{x}, \mb{k}(\mb{y}, \widehat{\mb{x}})) = L_\mb{f}\hphi(\mb{x}) + L_\mb{g}\hphi(\mb{x})\mb{k}(\mb{y}, \widehat{\mb{x}}) + \alpha (\hphi(\mb{x}) - \nu), \nonumber
    \end{equation}
    where we choose ${\nu = \mu}$ if $\hphi  = \hphi_{\tau_j}$ and ${\nu = 0}$ if $\hphi = \hphi_B$. It follows by Lipschitz continuity that: 
    \begin{align}
    \begin{split}
        \Vert L_\mb{f}\hphi(\widehat{\mb{x}}) - L_\mb{f}\hphi(\mb{x})\Vert_2 &\leq \mathfrak{L}_{L_\mb{f}\hphi}\epsilon(\mb{y}),\\
        \Vert \alpha(\hphi(\widehat{\mb{x}}) - \nu)  - \alpha(\hphi(\mb{x}) - \nu) \Vert_2 &\leq \mathfrak{L}_{\alpha} \mathfrak{L}_{\hphi} \epsilon(\mb{y}), \\
        \Vert L_\mb{g}\hphi(\widehat{\mb{x}}) - L_\mb{g}\hphi(\mb{x})\Vert_2 \Vert\mb{k}(\mb{y}, \widehat{\mb{x}}) \Vert_2 &\leq \mathfrak{L}_{L_\mb{g}\hphi} \epsilon(\mb{y}) \Vert \mb{k}(\mb{y}, \widehat{\mb{x}}) \Vert_2. \nonumber
    \end{split}
    \end{align}
    As $\mb{k}$ satisfies \eqref{eq:MRCBF_SDS}, we have that:
    \begin{align}
        c(\mb{x}, & \mb{k}(\mb{y}, \widehat{\mb{x}})) \nonumber \\
        & = c(\widehat{\mb{x}}, \mb{k}(\mb{y}, \widehat{\mb{x}})) + c(\mb{x}, \mb{k}(\mb{y}, \widehat{\mb{x}})) -  c(\widehat{\mb{x}}, \mb{k}(\mb{y}, \widehat{\mb{x}})) \nonumber \\
        & \geq c(\widehat{\mb{x}}, \mb{k}(\mb{y}, \widehat{\mb{x}})) - (a(\mb{y}) + b(\mb{y}) \Vert \mb{k}(\mb{y}, \widehat{\mb{x}}) \Vert_2) 
         \geq 0. \nonumber
    \end{align}
Since ${c(\mb{x}, \mb{k}(\mb{y}, \widehat{\mb{x}})) \geq 0}$ and $\mu$ satisfies \eqref{eq:varep}, we have that the system \eqref{eq:mu_ode} is safe with respect to $\mathcal{C}_I$ by \cite[Lemma 2]{gurriet2020scalable}.
\end{proof}

\begin{figure*}[!ht]
\centering
    \includegraphics[height=1.8in,width=0.23\linewidth]{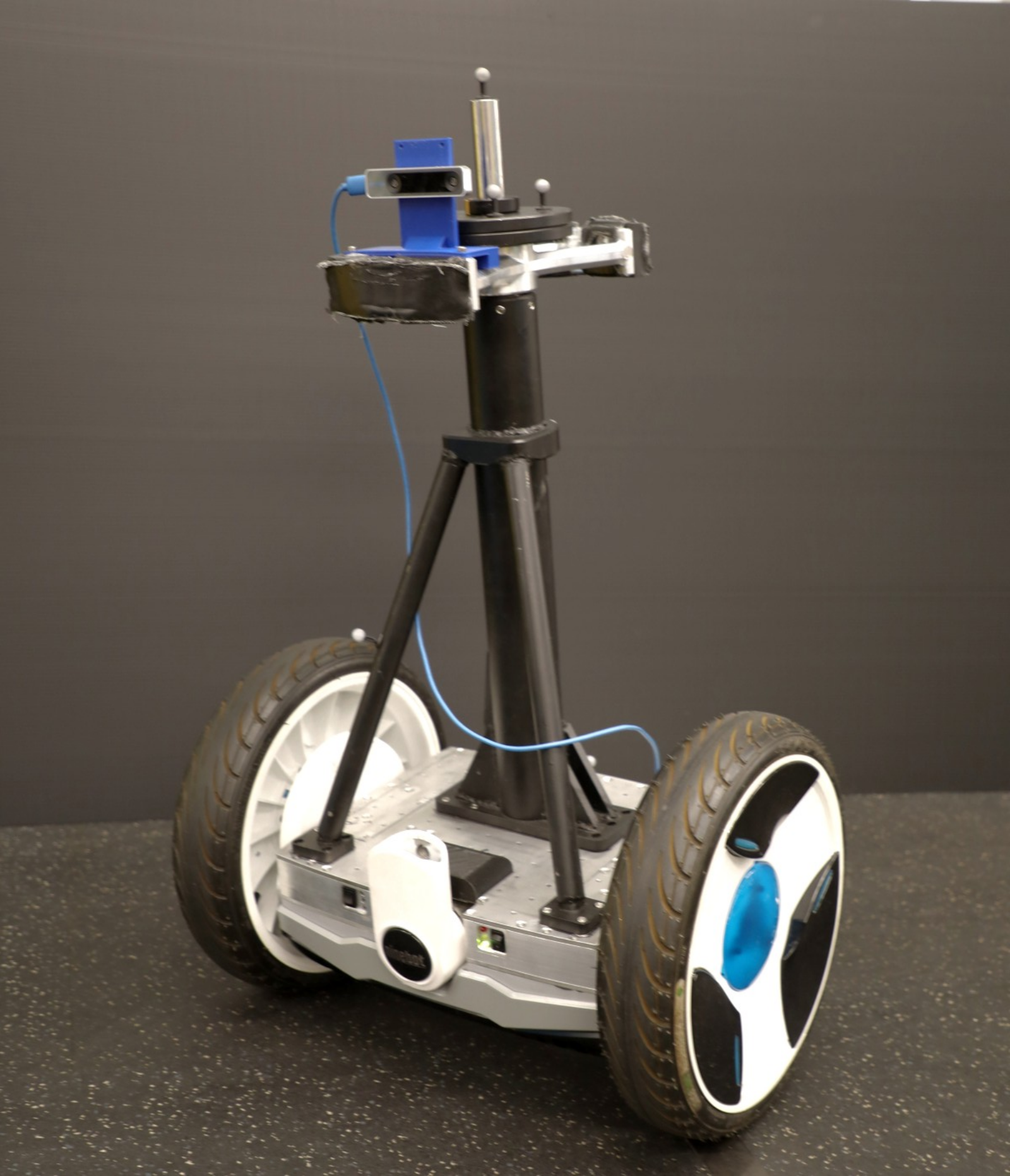}       \includegraphics[height=1.8in,width=0.37\linewidth]{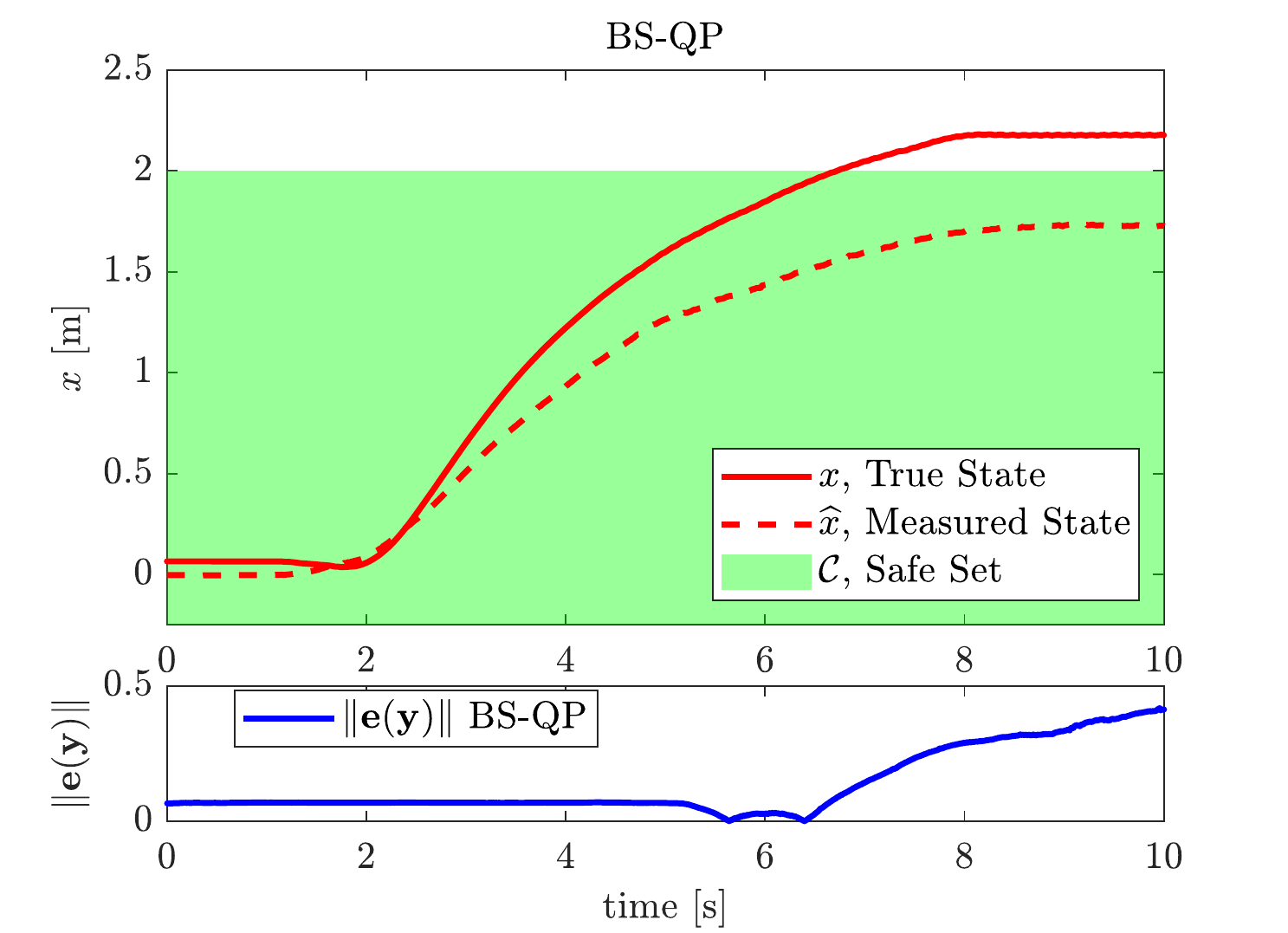}    
    \includegraphics[height=1.8in,width=0.37\linewidth]{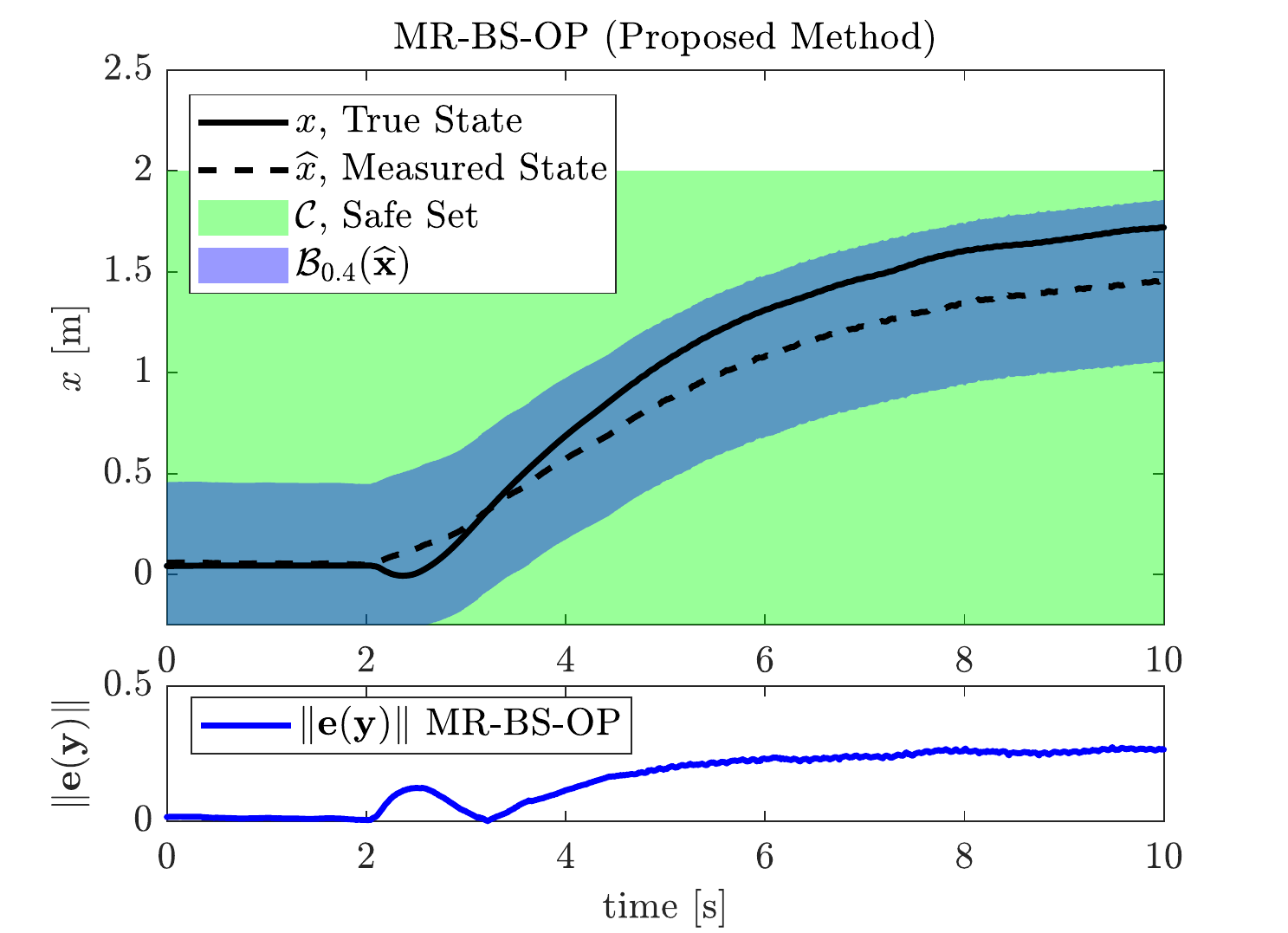}
    \caption{Experimental results using SLAM measurement model from the Intel RealSense T265 and constant desired velocity of 1 $\textrm{m/s}$. 
    \textbf{(Left)} An image of the physical Segway platform. \textbf{(Center)} Trajectories generated using the \ref{eqn:BS-QP}.
    Solid line represents the true state, dashed line shows the measured state, and green region indicates the safe set $\mathcal{C}$. The true trajectory fails to be safe and exits the safe set at $t=6.7$ s. The measurement error is plotted in blue. \textbf{(Right)} Trajectories generated using the \ref{eq:MRBS-OP}.
    An additional robustness region is plotted in blue to indicate the set of true states which the control input renders safe. Both the true and measured trajectories are safe demonstrating the robustness of the \ref{eq:MRBS-OP} when compared to the \ref{eqn:BS-QP}.}
     \label{fig:hardware}
     \vspace{-0.5cm}
\end{figure*}

This result allows us to present an alternative to the \ref{eqn:BS-QP} controller which adds the measurement-robustness of MR-CBFs. The constraints \eqref{eq:MRCBF_SDS} can be directly integrated into a Measurement-Robust Backup Set Optimization Program controller \ref{eq:MRBS-OP} as: 
\begin{align}
\label{eq:MRBS-OP}
\tag{MR-BS-OP}
\mb{k}(\mb{y},\widehat{\mb{x}}) =  \,\,&\underset{\mb{u} \in \R^m}{\argmin}    \quad \frac{1}{2} \Vert \mb{u}-\mb{k}_d(\widehat{\mb{x}}) \Vert_2^2 \nonumber \\
\mathrm{s.t.} &\quad  
L_\mb{f}\hphi_{\tau_j}(\widehat{\mb{x}}) + L_\mb{g}\hphi_{\tau_j}(\widehat{\mb{x}}) \mb{u} \nonumber\\
& \quad - (a_{\tau_j}(\mb{y}) + b_{\tau_j}(\mb{y})\Vert \mb{u} \Vert_2) \geq - \alpha( \hphi_{\tau_j}(\widehat{\mb{x}}) - \mu) \nonumber\\
& \quad L_\mb{f}\hphi_B(\widehat{\mb{x}}) + L_\mb{g}\hphi_B(\widehat{\mb{x}}) \mb{u} \nonumber\\
& \quad - (a_B(\mb{y}) + b_B(\mb{y})\Vert \mb{u} \Vert_2) \geq - \alpha ( \hphi_B(\widehat{\mb{x}})) \nonumber
\end{align}
for all $\tau_j \in \{0, \Delta_t, \dots, T \} $. This is a second-order cone program (SOCP), and there exists a wide array of solvers that are capable of implementing this controller including 
ECOS \cite{domahidi2013ecos}. Notably, the conservative nature of the method scales with the bound on the measurement-model error $\epsilon(\mb{y})$ and the \ref{eq:MRBS-OP} reduces to the \ref{eqn:BS-QP} when ${\epsilon(\mb{y}) =0}$.
We remark that the feasibility of \ref{eq:MRBS-OP} for all $\widehat{\mb{x}} \in \R^n$ can be ensured by adding a slack variable to the optimization problem.
The impact of the slack variable on safety can be understood via the concept of projection-to-state safety \cite{taylor2020control}.

\section{EXPERIMENTAL RESULTS} \label{sec:experiments}

In this section we demonstrate the efficacy of the proposed \eqref{eq:MRBS-OP} controller on a modified Ninebot E+ Segway platform in both simulation and experimentally on hardware. 

We consider a 4-dimensional asymmetrical Segway model shown in Figures \ref{fig:sim} and \ref{fig:hardware}. The state of the system consists of the position $x$, the forward velocity $\dot{x}$, the pitch angle $\psi$, and the pitch rate $\dot{\psi}$.
The equations of motion were derived using Newton-Euler method treating the Segway as an inverted pendulum with control input as torque command at the wheels; see \cite{gurriet2020scalable}.
The Backup Set method for generating control invariant sets is particularly relevant for this system due to its non-minimum phase dynamics. 

The desired safe set was chosen to be the set of states with position less than 2 m from the origin, i.e. $\mathcal{C} = \{ \mb{x} \in \R^n : x \leq 2\} $ and $h(\mb{x}) = 2 - x$. The backup controller was an LQR controller on the linearized system dynamics and the backup set was an estimate of the region of attraction of the LQR controller to the upright equilibrium state, given by a quadratic Lyapunov function. This set is then translated to match the current position of the Segway, while not allowing it to exceed the set boundary.
The functions $\overline{h}_\tau$, $\tau \in [0, T]$ were converted into four CBFs $\overline{h}_{\tau_j}$. Lastly, the Lipschitz constants for $\overline{h}_{\tau_j}$ were found explicitly by inspection of the Segway dynamics and the Lipschitz constants for $\overline{h}_B$ were found by sampling the state space in simulation and taking the largest numerical gradient. 

\begin{figure*}[!ht]
    \centering 
    \includegraphics[  width=0.98\linewidth]{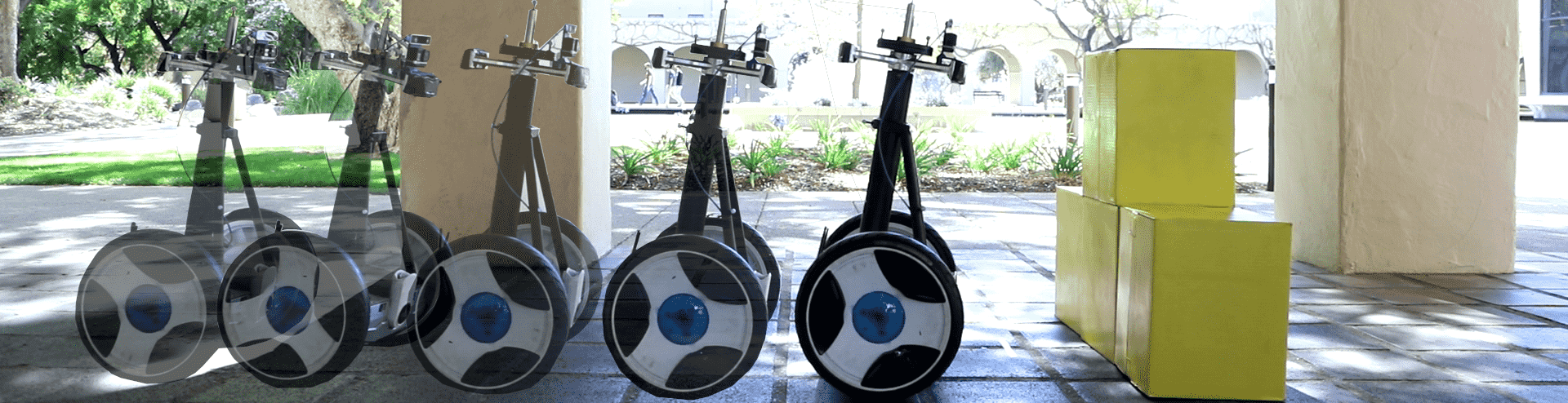}\\
    \vspace{0.1cm}
    \includegraphics[width=0.19\linewidth]{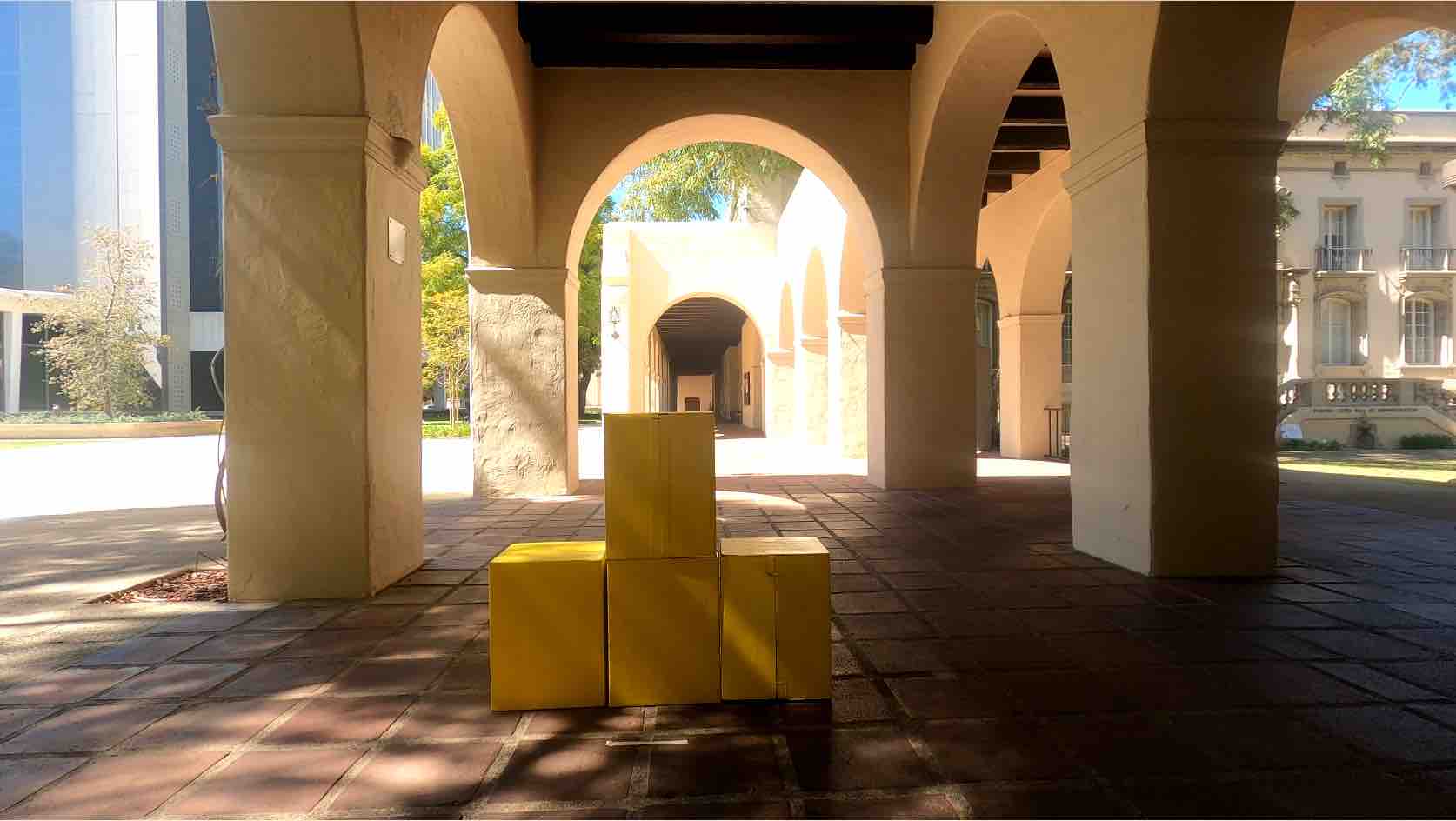}
    \includegraphics[width=0.19\linewidth]{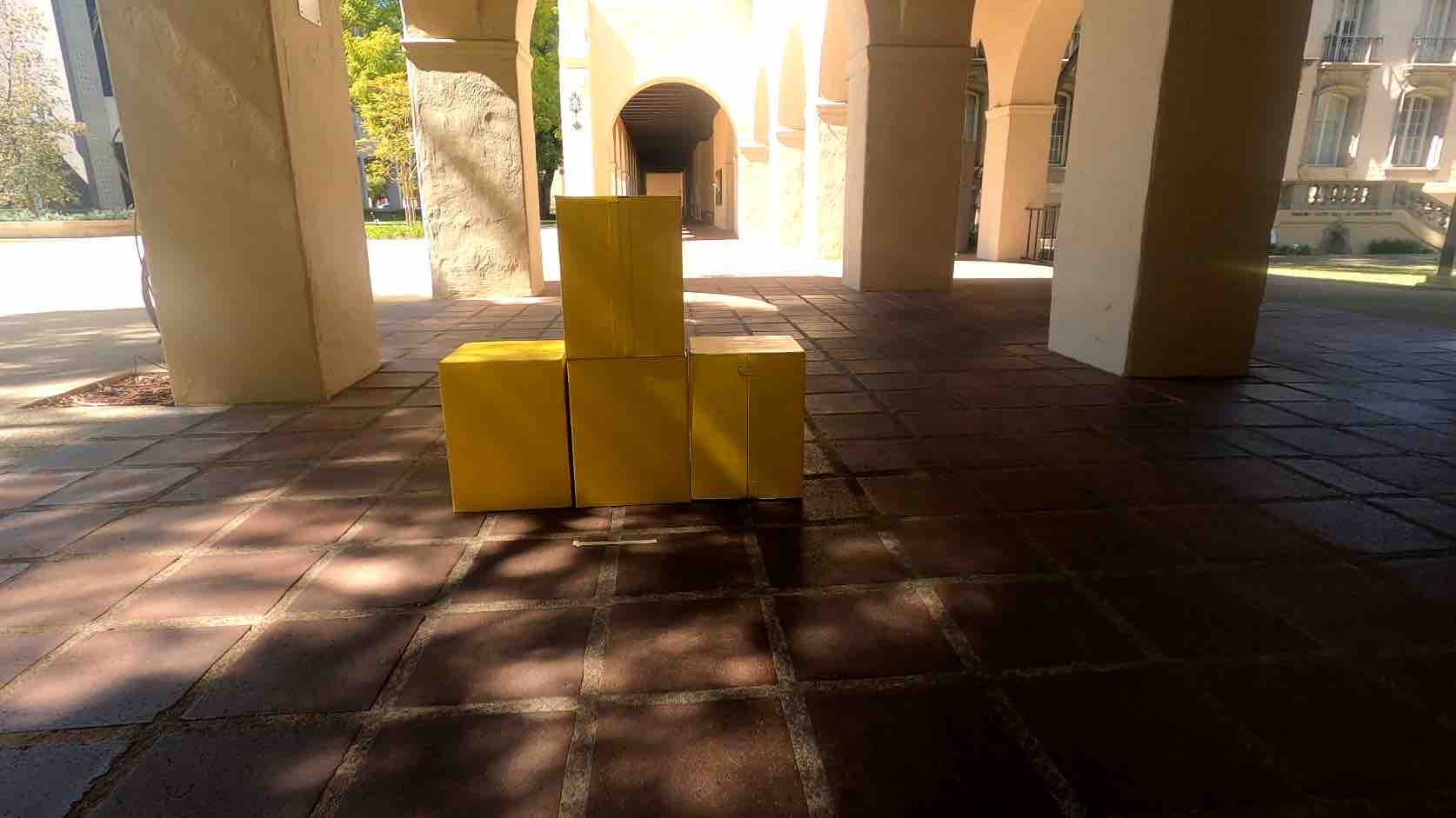}
    \includegraphics[width=0.19\linewidth]{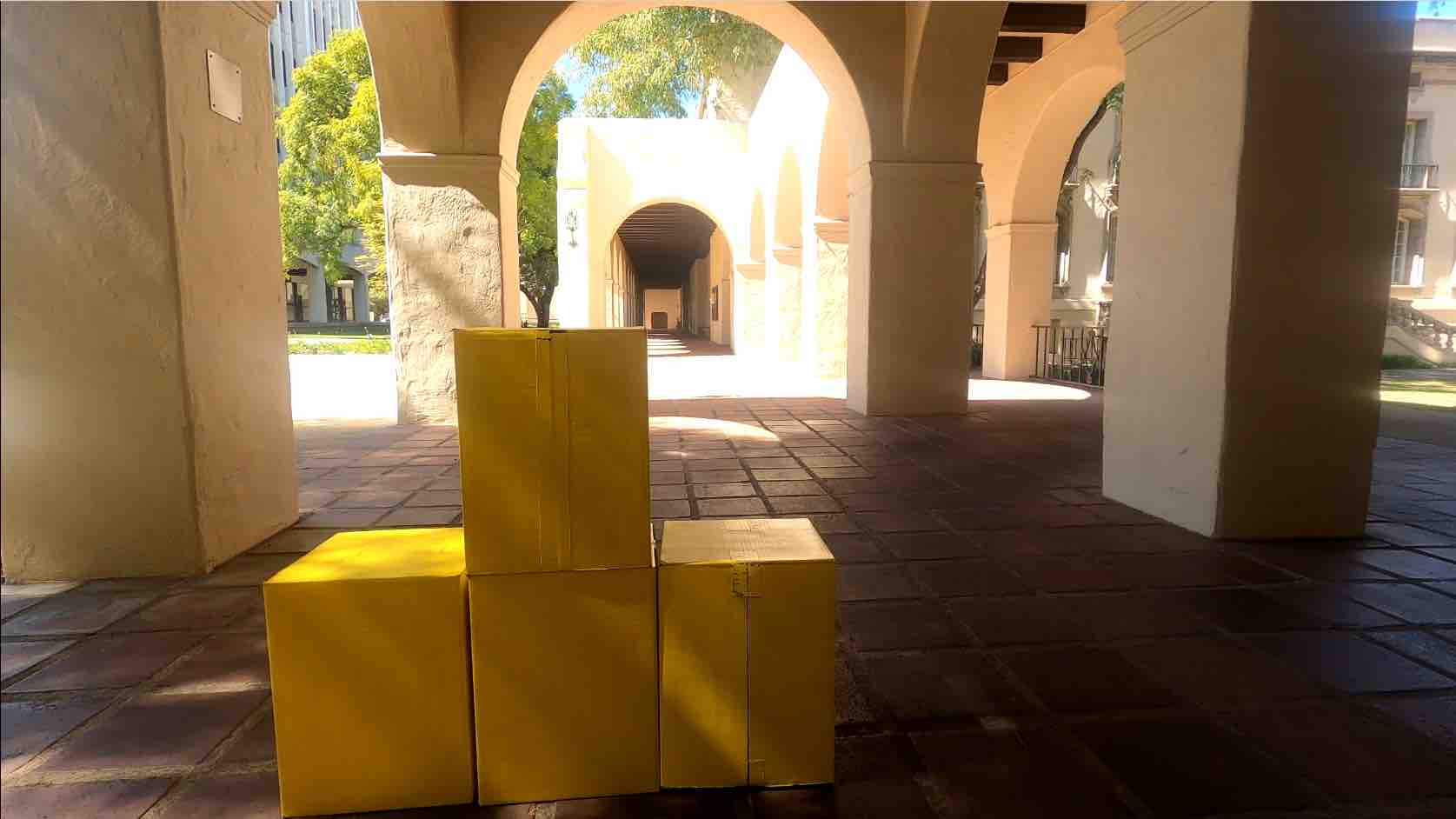}
    \includegraphics[width=0.19\linewidth]{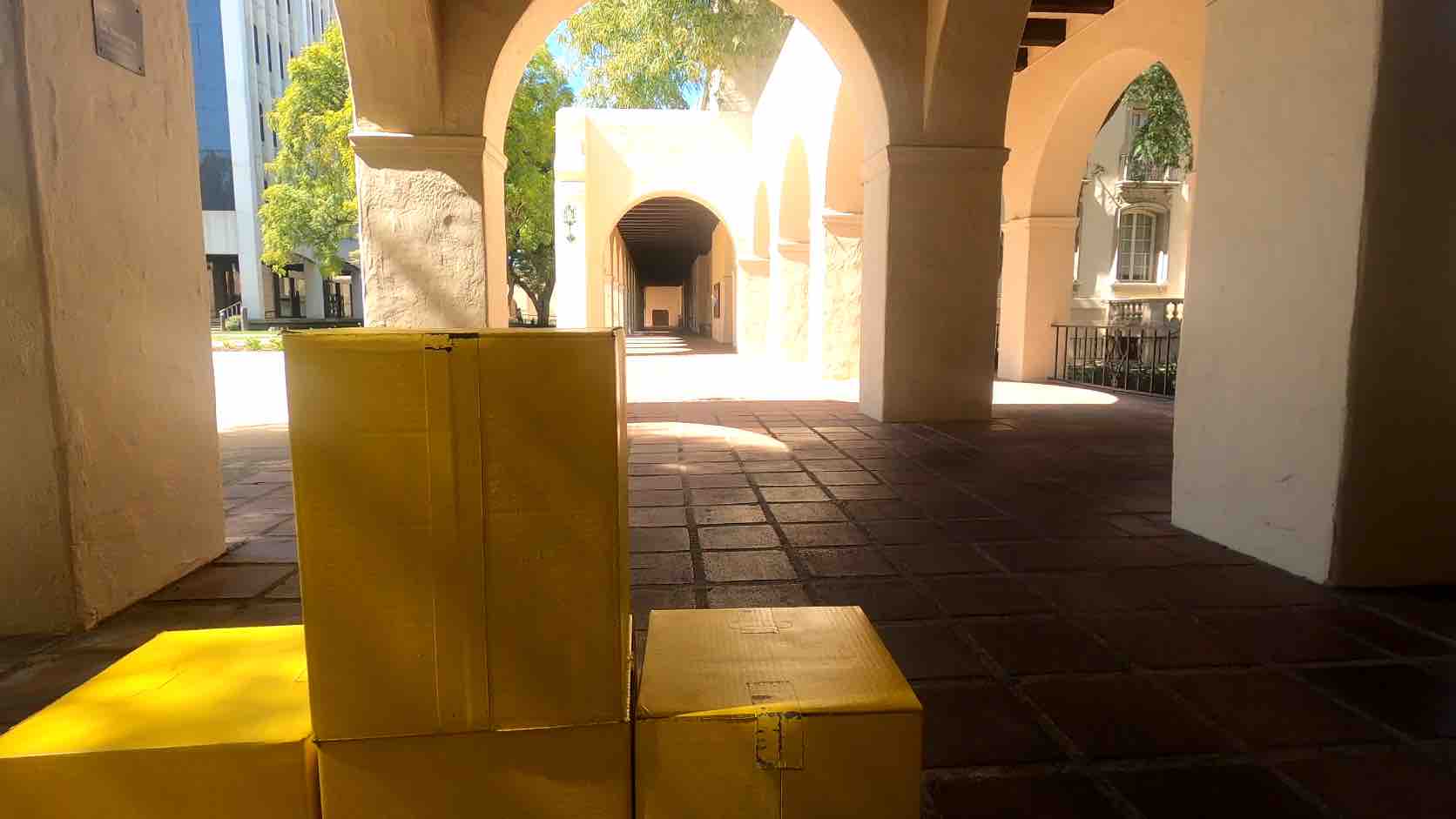}    
    \includegraphics[width=0.19\linewidth]{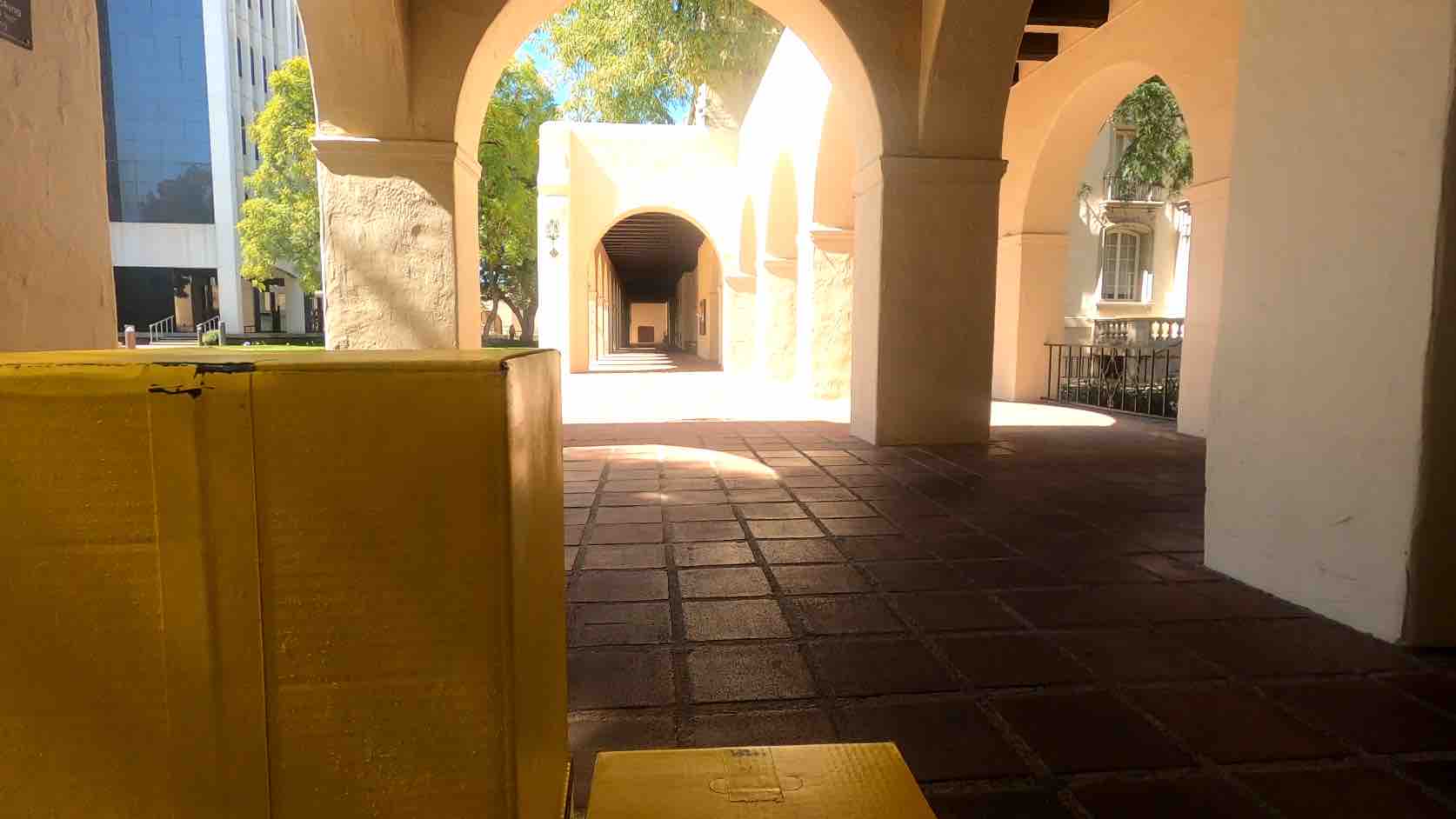}
    \caption{Images from the experiment using the \ref{eq:MRBS-OP} controller. The Segway is piloted towards a wall of yellow boxes and the controller ensures that it remains safe, i.e. that it does not crash into the boxes.  \textbf{(Top)} Time lapse of the Segway trajectory. \textbf{(Bottom)} Camera images taken from the perspective of the Segway across the experiment. The images are displayed in chronological order from left to right.}  \label{fig:outside}
    \vspace{-0.5cm}
\end{figure*} 

\vspace{0.1cm}
\noindent
\textbf{Simulation Results.}
The \ref{eq:MRBS-OP} was first validated in simulation in a ROS-based environment, found here\footnote{Simulation code \url{github.com/DrewSingletary/segway_sim}}. 
Measurement-model uncertainty was achieved by artificially adding a constant error of $-0.4$ m to the true state. The simple test scenario involved driving the Segway forward with a constant desired velocity of $1$ m/s. As seen in Figure \ref{fig:sim}, the \ref{eq:MRBS-OP} provided robustness to this error in this setting.  Importantly, without measurement-robustness, the system would be unsafe due to uncertainty in the state. 

\vspace{0.1cm}
\noindent
\textbf{Hardware Results.}
The \ref{eq:MRBS-OP} was then implemented on hardware. 
State estimates for $\dot{x}, \psi$, and $\dot{\psi}$ were found using wheel incremental encoders and a VectorNav VN-100 IMU.
The position estimate for $x$ was obtained from an Intel RealSense T265 camera running proprietary Visual Inertial Odometry (VIO) based SLAM. Onboard computation was performed by a Jetson TX2 which computes control actions and relays them to the low-level motor controllers. The TX2 concurrently runs Linux with ROS, enabling external communication and logging, and the ERIKA3 real-time operating system, which enables real-time low-level communication and computation of the control action. 

An OptiTrack motion capture system was used to provide state estimates which are considered true. These closely matched the encoder position estimates for short trials, so the encoder $x$ estimates are considered true in the outdoor experiments. This data was used to determine the 
error bound $\epsilon(\mb{y})$ that appears in the MR-CBF constraint. As the $(\dot{x}, \psi, \dot{\psi})$ state estimates provided by the encoders and IMU are highly accurate, we focus on making the system robust to measurement error in its vision-based position estimate $\widehat{x}$. 

The value $\epsilon(\mb{y}) =0.4$ was chosen as an upper bound on the measurement error for all $\mb{y} \in \mb{p}(\mathcal{C})$. 
The \ref{eq:MRBS-OP} was implemented at the embedded level in the ERIKA3 operating system using the ECOS \cite{domahidi2013ecos} SOCP solver. The desired controller $\mb{k}_d$ was a proportional-derivative controller tracking user velocity inputs. The backup trajectory $\phi_\tau^{\mb{k}_B}(\widehat{\mb{x}})$ and its partial derivatives were approximated via Euler integration using a time step of $\Delta t = 5$ ms and the time used to expand the backup set $\mathcal{C}_B$ to $\mathcal{C}_I$ was $T=1$~s. 
The \ref{eq:MRBS-OP} ran at 250 Hz with 5 decision variables, 4 linear constraints, and 6 second order cone constraints and saturated at $\pm 20~ \textrm{N}\textrm{m}$. 

To demonstrate the method, a simple scenario is executed on the Segway in which it is 
driven forward at its maximum velocity of 1 m/s. This scenario is performed with both the \ref{eqn:BS-QP} and the \ref{eq:MRBS-OP}. The results of these experiments can be found in Figure \ref{fig:hardware},  images from the experiment can be seen in Figure \ref{fig:outside}, and a 
video can be found at \cite{videoVideo}. With the \ref{eqn:BS-QP} controller the estimated state $\widehat{\mb{x}}$ remains safe, but the true state $\mb{x}$ becomes unsafe whereas with the \ref{eq:MRBS-OP} controller
both the estimated and the true state are kept safe.
This highlights the importance of providing robustness against measurement uncertainty, as achieved by Theorem~\ref{thm:main}. 




\section{CONCLUSION}
\label{sec:conclusion}

This paper established robust controller synthesis with formal safety guarantees for systems relying on uncertain measurements.
We approached this problem through the framework of CBFs.
We additionally highlighted the importance of control invariant sets and experimentally implemented the Backup Set method to produce such a set for a Segway.
Our theoretical construction culminated in the integration of the Backup Set method with MR-CBFs, which provides robustness to state measurement uncertainty in the safety guarantees.
We implemented the proposed control method on a Segway platform and demonstrated robustly safe operation in experiments. 
Future work includes addressing feasibility of the \ref{eq:MRBS-OP} for general systems in the presence of probabilistic error bounds and developing methods to determine the Lipschitz constants in its constraints.











\bibliographystyle{IEEEtran}
\bibliography{cosner_main}

\end{document}